\documentclass[USenglish,cleveref,autoref]{lipics-v2021}
\newtheorem{obs}[theorem]{Observation}
\def\F{{\cal F}}
\def\S{{\cal S}}
\def\N{{\bf N}}
\def\mapfam{-mapping family }
\def\im{{\rm Im}}

\title{Uniform Universal Sets, Splitters, and Bisectors} 

\author{Elisabet Burjons}{Serra H\'{u}nter Fellow, Universitat Polit\`{e}cnica de Catalunya}{elisabet.burjons@upc.edu}{https://orcid.org/0000-0001-6161-7440}{}

\author{Peter Rossmanith}{Department of Computer Science, RWTH Aachen,
Germany}{rossmani@cs.rwth-aachen.de}{https://orcid.org/0000-0003-0177-8028}{}

\authorrunning{E.~Burjons and P.~Rossmanith}

\Copyright{Elisabet Burjons and Peter Rossmanith}



\nolinenumbers

\begin{document}

\maketitle

\begin{abstract}
Given a subset of size $k$ of a very large universe a randomized way
to find this subset could consist of deleting half of the universe and
then searching the remaining part.  With a probability of $2^{-k}$ one
will succeed.  By probability amplification, a randomized algorithm needs
about $2^k$ rounds until it succeeds.  We construct \emph{bisectors} that
derandomize this process and have size~$2^{k+o(k)}$.  One application is
derandomization of reductions between average case complexity classes.
We also construct \emph{uniform $(n,k)$-universal sets} that generalize
universal sets in such a way that they are bisectors at the same time.
This construction needs only linear time and produces families of
asymptotically optimal size without using advanced combinatorial
constructions as subroutines, which previous families did, but are based
mainly on modulo functions and refined brute force search.
\end{abstract}

\clearpage
\section{Introduction}

There are several good reasons why one might prefer deterministic
algorithms over probabilistic ones.  An example is to avoid
having one-sided errors and another one is to have guaranteed
running times.  Therefore derandomization of probabilistic
algorithms is an important topic.  Unfortunately, one pays
a penalty in the running time for having a deterministic
algorithm~\cite{BanikPRS2018,DGHKKPRS2002,HSSW2002,Schoning2002,DantsinGHKKPRS2002,MoserS2011}.

Kleitman and Spencer introduced universal sets~\cite{KleitmanS1973},
which later turned out to be an important tool in derandomization.  Naor,
Schulman, and Srinivasan~\cite{NaorSS95} improved the construction of
universal sets and introduced splitters to derandomize
algorithms in such a way that leads to determistic algorithms
whose running time is almost the same as their randomized
counterparts'~\cite{BlasiokK2017,CaiCC2006,CKLMRRSZ2009,JansenLR2020,NaorSS95}. 

Bisectors were introduced in~\cite{DreierLR2020} as a way to
derandomize reductions between average case complexity classes. 

An $(n,k,\ell)$-splitter is a family of functions $f\colon[n]\to[\ell]$,
such that for every $k$-subset $S$ of $[n]$, there is a function~$f$ in
the family splitting $S$ evenly or as evenly as possible into parts
$f^{-1}(x)\cap S$ for $x\in[\ell]$ (we write $[n]=\{1,\ldots,n\}$).
If $k\le \ell$ this means that every element of $S$ is mapped to a
different element in $[\ell]$.  In particular,
$(n,k,k)$-splitters are perfect hash functions, which can be used,
for instance, to derandomize color coding~\cite{AlonYZ1995}.

Universal sets are a more general version of $(n,k,2)$-splitters.
An $(n,k)$-universal set is a family of functions $[n]\to\{0,1\}$
such that for every possible
way of dividing a $k$-subset $S$ into two subsets,
there is a function mapping one of
them to~$0$, and the other to~$1$.  An early motivation
is testing of circuit components where each component relies on at
most $k$ inputs. A lower bound of $\Omega(2^k\log n)$ exists for their
size~\cite{KleitmanS1973}. The existence of such sets is relatively easy
to prove by using the union bound and probabilistic arguments.

Finding such universal sets and splitters through brute force takes too
long, and Naor et al.~\cite{NaorSS95} use a combination of exhaustive
search on probability spaces, and advanced combinatorial concepts
like error correcting codes to build almost optimal splitters
and universal sets in time linear in the output size.

A bisector is, in a way, a degenerated universal set, where for every $k$-subset $S$ 
there must be a function in the bisector mapping $S$ completely to $0$. One
could think that such families are not interesting,
because only
one function, which maps every element to $0$, would already fulfill this
condition.  However, we also have a global condition, which we have seen
neither in splitters nor in universal sets so far. An $(n,k)$-bisector
is a family of functions $f\colon[n]\to\{0,1\}$ such that \emph{at least one} function
maps a $k$-subset $S$ to $0$ for every possible $S$, and \emph{every}
function maps exactly half of $[n]$ to $0$ and the other half to $1$.

Universal sets, bisectors and splitters are related notions, with one striking difference.
In a bisector we have a global condition on each function, whereas there is no such restriction
in the case of universal sets or splitters. However, the probabilistic method suggests
that there should exist small universal sets and splitters, even if we require that the functions
have the same global property as bisectors.

Motivated by the global property of uniformity that bisectors have, where every function maps half the elements to 0 and the other half to 1, in this paper we introduce the notion of uniformity for splitters and universal sets, by requiring each function to have a uniform mapping to its image, and construct almost optimal uniform splitters and universal sets, as well as almost optimal bisectors.

Requiring such global properties might also be useful in practice. As we have already mentioned, an early  motivation of universal sets is 
testing of circuit components where each component relies on at
most $k$ inputs.
In a setting where not too many inputs of the circuit are allowed to
receive a $1$-input traditional universal sets cannot be used.
Uniform universal sets can do the job and they provide an optimal
tradeoff between the number of necessary tests and the allowed hamming
weight of the test inputs.  In particular if half of the inputs are
allowed to be ``active,'' then the number of tests is asymptotically not
bigger than the required number without restriction on the hamming
weight.

Uniform splitters also have potentially more applications than
the original ones, thanks to their additional properties.
For instance, one can use an $(n,k,\ell)$-splitter
to design a basic secret sharing scheme, but with a uniform one you
can guarantee that secrets are not shared with too many parties, which
is a potential security risk.

There are also some applications for bisectors.  The \emph{even set problem}
essentially asks whether an $n\times n$ matrix over ${\bf F}_2$ contains
$k$ row vectors that are linearly dependent (over ${\bf F}_2$).
Only recently it was shown that this problem is $\rm
W[1]$-hard~\cite{BhattacharyyaGS18}.  What is the average case
complexity of this problem?  If we look at a random $\frac n2\times
n$ matrix then the problem is easy to solve:  With overwhelming
probability such a matrix has full rank and only if it does not we
check whether the $k$ row vectors exist.  As this has to be done
only very rarely the expected running time is very fast.  We have,
however, an $n\times n$ matrix and this trick does not work because
with a relatively high probability, the rank is less than~$n$.  We can
still solve the problem efficiently by using a bisector and reducing
the problem for square matrices to the problem for $\frac n2\times n$
matrices. Other examples for
reductions between average case problems that use bisectors can be
found in~\cite{DreierLR2020}.

Another application is \emph{parallelization of black-box search
algorithms.}  If the exponential time hypothesis holds, then finding a
$k$ vertex guest graph as a subgraph in an $n$-vertex host graph cannot
be done in time~$n^{o(k)}$~\cite{CFGKMPS2016}, but of course it can
be done in time~$O(n^{k})$.  Can, thus, an algorithm be parallelized
\emph{without changing the algorithm itself?} Yes, we can apply a
bisector to the vertex set of the host graph and get $2^{k+o(k)}$ many
new host graphs of size $n/2$.  We check all of them in parallel in
time~$O((n/2)^{k})=O(2^{-k}n^k)$ using $2^{k+o(k)}$ processors giving
us an asymptotically optimal speedup.


The notion of splitters and universal sets with global properties is not new.
In a \emph{balanced} splitter introduced by Alon and Gutner,
one requires that every $k$-subset is split by about the same number of functions. 
Balanced splitters can be used for approximate counting~\cite{AlonG2009}.
The notion of balanced splitters is natural in the sense that if one would construct
a splitter at random, one would expect the splitter to have this property.

Thus, our constructions seek to answer a very natural question. Can we
build small deterministic splitters and universal
sets which maintain the same properties we would expect from their randomly built counterparts?
Our answer is partially yes, if one focuses on uniformity. It remains open, however, if one could build
such families to be both uniform and balanced.

Another advantage to our construction is its simplicity, our constructions
use a combination of modulo functions and total enumeration. This
makes these constructions easier to implement than the previously best
known splitters and universal sets from~\cite{NaorSS95}, which rely on
assymptotically good error correcting codes~\cite{AlonBNNR1992}.
The price we pay for this simplicity is a more complicated analysis of
the sizes of our splitters, bisectors, and universal sets.

\subsection{Our Contributions}

As already mentioned, splitters and universal sets were made almost optimal by
Naor et al.~\cite{NaorSS95}.
For instance, they presented an $(n,k)$-universal set of size 
$2^kk^{O(\log k)}\log n$ asymptotically matching the lower bound, while
the previous best universal sets had size 
$O(\min\{k2^{3k}\log n, k^22^{2k}\log^2 n\})$~\cite{AlonBI1986}.
This is not the case for bisectors, maybe because they are a fairly new concept. 
The size of an $(n,k)$-bisector can be easily lower
bounded by $2^k$.
Given an $(n,k)$-bisector $\F$, one can always choose an $x\in[n]$ such that
$f(x)=1$ for at least half of the functions in $\F$. We can construct a set
$S$ containing $x$, and half of $\F$ will not map~$S$ to $0$.
Adding a new element to~$S$ with the same criterion will again halve the candidate functions
for appropriately mapping~$S$. Repeating this procedure $k$ times we realize that any
bisector must contain at least $2^k$ functions.
Unlike in the case of universal sets, the best known bisectors until now have size
$4^k$~\cite{DreierLR2020} using a very straightforward construction.
In this paper, we show that we can construct a bisector of size $2^{k+o(k)}$ in linear time.
It is important to note that these bisectors have sizes independent of $n$, while the size
of a universal set is at least logarithmic in~$n$.

In order to build small bisectors, we use the following strategy depicted 
in the left side of~\cref{fig:bisect}. The idea is that we build a family of functions
$f\colon[k^3]\to\{0,1\}$,
where every $k$-subset of $[k^3]$ is mapped to $0$, and every function maps only $k^3/\sqrt{k}$
elements to $1$, and the rest to $0$. This is more or less easy to do because the number of $1$s
in every function is relatively small, so the chance of finding a good function is large.
Then, we repeat this process many times, every time mapping a fraction of $1/\sqrt{k}$ of the 
remaining $0$s to $1$, until we have a $(k^3,k)$-bisector.
We expand this to an $(n,k)$-bisector using only modulo functions and the Chinese remainder 
theorem, with a size blowup of only a factor of $k$. Lastly, we can speed up the process by
using this construction on subintervals of $[n]$ containing only a few elements of a candidate
set $S$.

Not only is this construction able to build $(n,k)$-bisectors of size $2^{k+o(k)}$ in linear
time, but for any constant $0<\alpha< 1$, we can build families of functions 
$f\colon[n]\to\{0,1\}$ that map every $k$-subset to $0$ and every function in the family
maps exactly a fraction of $\lceil\alpha n\rceil$ elements to $1$ and the rest to $0$.
We call these $(n,k,\alpha)$-bisectors, and we build such bisectors of size
$(1/(1-\alpha))^{k+o(k)}$ in linear time.

If one builds an $(n,k,\ell)$-splitter built at random, one expects that
that each function $f$ maps $[n]$ into $\ell$ parts of almost equal size.
We say that an $(n,k,\ell)$-splitter is \emph{uniform} if it has this property.
In this paper, we build such splitters for $\ell\ge k^3$,
with size $O(k^6\log n)$. These splitters can be used just like the splitters of Naor et al.~\cite{NaorSS95}, and one can additionally take advantage of their uniformity.

The same argument holds when we talk about uniform universal sets
instead of splitters.
If one constructs a random universal set,
one expects to obtain functions which map
more or less half of the elements to $0$ and the other half to $1$.
However, in known determistic constructions this is not the case.
We use uniform splitters as the basis to construct uniform universal sets.
Building uniform universal sets for large values of $n$ is not easy,
but, just as in the bisector case,
we start by building uniform sets for $n=k^3$ 
which maps a fraction of any $k$-subset to $1$, as in the right side
of~\cref{fig:bisect}.  Then, we make the size bigger by combining a
uniform $(n,k,k^3)$-splitter with one such uniform set. Here, one needs
to be careful about respecting the uniformity. Finally we can take for
every possible subdivision of a subset into two parts a uniform set
that will map them appropriately. The union of all those sets will be a
uniform $(n,k)$-universal set. Here again, just as with the bisectors,
we can be flexible about the uniformity of the functions and build for any
constant $0<\alpha\le 1/2$ a uniform $(n,k,\alpha)$-universal set, where
every function maps exactly $\lceil\alpha n\rceil$ elements to $1$ and
the rest to $0$. This is one of the main results of our paper.  Not only,
do these sets have size $(1/\alpha)^{k+o(k)}\log n$, but they can also
be built in linear time. This means that they can be used everywhere
you can use the
universal sets by Naor, Schulman, and Srinivasan~\cite{NaorSS95},
but they are additionally uniform.

The paper is structured as follows. First, we present uniform $(n,k,\ell)$-splitters for $\ell\ge k^3$.
Then, we present bisectors whose size is independent of $n$, and finally, we build uniform universal sets.


\section{Uniform \mbox{\boldmath$(n,k,\ell)$}-splitters}\label{sec:splitters}

A splitter is a family of functions from $[n]$ to $[\ell]$ such that for every $k$-subset there is
a function that distributes its elements evenly. Formally:

\begin{definition}
Let $n$, $k$, and $\ell\le n$ be integers. An \emph{$(n,k,\ell)$-splitter} is a family $\F$ 
of functions $[n]\to [\ell]$ such that for every subset $S\subseteq [n]$ with $|S|=k$,
there is an $f\in\F$ that splits $S$ into $\ell$ parts 
$\lfloor k/\ell \rfloor \le |f^{-1}(i)\cap S|\le \lceil k/\ell \rceil$
for $i=1,\hdots,\ell$. 
\end{definition}

Observe, that if $\ell\ge k$, then given $i,j\in S$ with $i\neq j$,
it is enough if $f(i)\neq f(j)$, and in that case it 
is not necessary that $\im(f)=[\ell]$ for every $f\in \F$.
Constructing a splitter consists on choosing one such set of functions and listing 
them as a table. This means that the time complexity of building a splitter is $\Omega(n|\F|)$.

For $\ell=k^2$,  Naor et al.~\cite{NaorSS95} show that 
there is an $(n,k,k^2)$-splitter of size $O(k^6 \log k \log n)$.
This family was obtained using asymptotically good error correcting codes, and if
we apply the same error correcting codes to obtain a splitter with $\ell=k^3$, we 
can obtain an $(n,k,k^3)$-splitter of size $O(k^4\log k\log n)$.

Our goal is to build splitters of similar size in a way that
does not use such elaborate constructions. In fact, our construction only uses
modulo functions and brute force search,
and the splitters we construct are additionally uniform. 
Let us define what we mean by uniform
splitters.

\begin{definition}
 Let $n$, $k$, and $\ell$ be integers. A \emph{uniform $(n,k,\ell)$-splitter} is an $(n,k,\ell)$-splitter $\F$ such that 
 for each $f\in\F$ and every $i\in \im(f)$, $\lfloor n/|\im(f)|\rfloor\le |f^{-1}(i)|\le \lceil n/|\im(f)|\rceil$.
\end{definition}

A relaxed notion of this includes the possibility that there is some
unevenness in the preimage.

\begin{definition}
 Let $n$, $k$, $\ell$, and $a$ be integers. An \emph{$a$-uniform $(n,k,\ell)$-splitter} is an $(n,k,\ell)$-splitter $\F$ such that 
 for each $f\in\F$ and every $i\in \im(f)$, $\lfloor n/|\im(f)|\rfloor-a\le |f^{-1}(i)|\le \lceil n/|\im(f)|\rceil+a$,
 and for every $i,j\in\im(f)$,
 ${\bigm|}|f^{-1}(i)|-|f^{-1}(j)|{\bigm|}\le a$.
\end{definition}

So the difference between the sizes of two preimages is never greater than $a$.
A final notion of uniformity can even be stronger, if the image of every
function must be $[\ell]$.

\begin{definition} 
 Let $n$, $k$, and $\ell$ be integers. A \emph{strongly uniform $(n,k,\ell)$-splitter} is an $(n,k,\ell)$-splitter $\F$ such that 
 for each $f\in\F$ and every $i\in[\ell]$, $\lfloor n/\ell\rfloor\le |f^{-1}(i)|\le \lceil n/\ell \rceil$.
\end{definition}

From now on, we assume that we talk about splitters where $\ell\ge k$, unless otherwise stated.
There are some ways to strengthen the uniformity of a splitter without making it too much bigger.

\begin{lemma}[Smoothing Lemma]\label{lem:smooth}
Let $\cal F$ be an $a$-uniform $(n,k,\ell)$-splitter.  If $n\geq a\ell(k+1)$
then we can construct from $\cal F$ a uniform $(n,k,\ell)$-splitter $\cal
F'$.  The time to compute $\cal F'$ from $\cal F$ is linear and
$|{\cal F'}|=(k+1)|\cal F|$.  
\end{lemma}

\begin{proof}
Let $f\colon[n]\to[\ell']$ be a function from~$\cal F$ for some $\ell'\le \ell$.  We can construct
a function $g\colon[n]\to[\ell']\times{\bf N}$ with the following
properties:  (1)~While $f$ is clearly not injective, $g$ is bijective.
(2)~If $g(i)=(x,y)$ then $f(i)=x$.  (3)~If $I=f^{-1}(i)$ then
$g(I)=\{(i,0),(i,1),(i,2),\ldots,(i,|I|-1)$.  We can use $g(i)=(x,y)$
as the coordinates in two dimensional table into which we map the
elements of~$[n]$.  This table has $\ell'$ columns.
Figure~\ref{fig:smooth} shows an example of such a table.
\begin{figure}
\centerline{\includegraphics{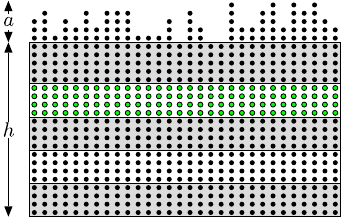}\hfil
            \includegraphics{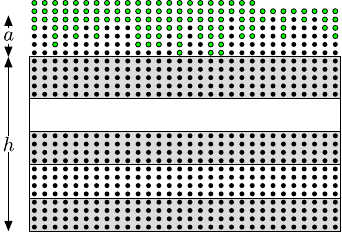}}
\caption{An example of a $5$-uniform $(n,k,\ell)$-splitter with
$n=713$, $k=5$, $\ell=30$.  On the left side one function maps the
$n$ numbers into the $l$ columns.  Due to the $5$-uniformity the
heights of the columns differ by at most~$5$.  We
construct a new function by choosing the green stripe and
redistributing its elements.
If the green
stripe is clean, then the new function is injective on~$S$.
}
\label{fig:smooth}
\end{figure}

Because $\cal F$ is $a$-uniform the number of elements in each column
is nearly the same -- their numbers differ at most by~$a$.  Let us
denote the minimum number of elements in a column by~$h$.  It is clear
that $\ell'(h+a)\geq n\geq \ell'h$.  We divide the table into $k+1$ stripes.
The first $k$ stripes have together height~$h$ which is distributed
equally among them.  Hence, the height of a stripe is $\lfloor h/k\rfloor$
or~$\lceil h/k\rceil$.  The last stripe contains all elements above~$h$.

Let us now assume that $f$ splits a set $S\subseteq[n]$, $|S|=k$, in
such a way that every element from $S$ is mapped to a different column.
There must be at least one stripe that does not contain an element of~$S$.
We call such stripes \emph{clean} and the others \emph{dirty}.  We
construct $k+1$ new functions that behave mostly like $f$, but redistribute
the elements of one of the stripes to make the sizes of all columns
equal or almost equal.  This is always possible for the top-most stripe, but
it is not guaranteed for the rest of the stripes. However, it is enough
if a middle stripe contains at least $(a-1)\ell'$ elements.
As we have already mentioned, a middle stripe will contain at least
$\lfloor h/k\rfloor\ell'$ elements, so we only need to prove that
$\lfloor h/k\rfloor\ge a-1$. We know that $\ell'(h+a)\geq n$,
and by the precondition we also know that $n\geq a\ell(k+1)\geq a\ell'(k+1) $.
If we put the two together this means that $h\ge ak$, 
which in its turn means that $\lfloor h/k\rfloor\ge h/k -1\ge  a-1$.
%
\end{proof}

%
%
%

We build our splitters based on the following properties of the modulo functions.
We know that $a-b\equiv0 \pmod m$ if and only if $a=b \bmod m$. Moreover, the modulo function
also behaves well with the product in the following sense, if $a\equiv0 \pmod m$ or $b\equiv0 \pmod m$
then $ab\equiv0 \pmod m$.  If $m$ is a prime number, then the converse is also true.
Finally, we also observe that modulo functions are uniform.

We also use the following observation.
For every subset $S$ of size $k$ in $[n]$, we can select all of the differences between elements in $S$,
there are ${k\choose 2}$ of them, and multiply them
together to get a number smaller than $n^{k^2/2}$. If we select enough
modulo functions of prime numbers between 
$k^2$ and $\ell$ we can guarantee by the Chinese remainder theorem
that none of the differences will evaluate to $0$
for some of the selected modulo functions. However, guaranteeing that there are enough prime numbers to perform this
trick is not completely straightforward, as we are going to see.

First we need the following observation, which follows from a
variant of the prime number theorem~\cite{RS62}.
\begin{obs}\label{obs:primenumthm}
 Given $k$ and $n$ such that $n\ge k\ge2$,
 \[\pi(k^2\log n)-\pi(k^2\log n/2)\geq \lceil k(k-1)\log n/(4\ln k+2\ln\log n-2\ln2))\rceil.\]
\end{obs}

\begin{proof}
 First we approximate the first two terms of the inequality using the following variant of the prime number theorem.
 For every $x\geq59$ Rosser and Schoenfeld~\cite{RS62} show that
\[
{x\over\ln x}\Bigl(1+{1\over 2\ln x}\Bigr)<\pi(x)
<{x\over\ln x}\Bigl(1+{3\over 2\ln x}\Bigr).
\]
So, we can approximate $\pi(k^2\log n)-\pi(k^2\log n/2)$ as:
 \begin{align*}
 \frac{k^2\log n}{\ln (k^2\log n)}\Bigl(1+\frac{1}{2\ln (k^2\log n)}\Bigr)
 &-\frac{k^2\log n}{2\ln (k^2\log n)-2\ln 2}\Bigl(1+\frac{3}{2\ln (k^2\log n)-2\ln 2}\Bigr)\\
 &\ge\frac{k^2\log n}{2\ln (k^2\log n)-2\ln2}-\frac{k\log n}{2\ln (k^2\log n)-2\ln2}\;.
 \end{align*}
 We can immediately take $k\log n$ as a common factor out of the equation, as it is never $0$. We also name $t=\ln (k^2\log n)$
 \begin{align*}
 \frac{k}{t}\bigl(1+\frac{1}{2t}\bigr)
 -\frac{k}{2(t-\ln2)}\bigl(1+\frac{3}{2(t-\ln2)}\bigr)\ge
 \frac{k}{2(t-\ln2)}-\frac{1}{2(t-\ln2)}\;.
 \end{align*}
 Now, we can easily multiply the whole inequality by $4t^2(t-\ln2)^2>0$ to get rid of the denominators and leave
 the whole inequality as a polynomial in $t$.
 \begin{align*}
  &4kt(t-\ln 2)^2+2k(t-\ln 2)^2-2kt^2(t-\ln 2)+3kt^2\ge 2kt^2(t-\ln 2)-2t^2(t-\ln 2)\\
  &4kt(t-\ln 2)^2+2k(t-\ln 2)^2-(4k-2)t^2(t-\ln 2)+3kt^2\ge 0\\
  &4k(t^3-2\ln2 t^2+(\ln2)^2t)+2k(t^2-2\ln2t+(\ln2)^2)-(4k-2)(t^3-\ln2 t^2)+3kt^2\ge0\\
  &2t^3+(5k-4k\ln2-2\ln2)t^2+4k((\ln2)^2-\ln2)t+2k(\ln2)^2\ge0\;.
 \end{align*}
 It is immediate to see that the last inequality holds for any $k\ge2$ as the only negative
 term is the one multiplied by $t$ and the terms in $t^3$ and $t^2$ dominate it, for instance if
 $t\ge 4$ the term in $t^2$ is enough. 
\end{proof}

\begin{lemma}\label{lem:s1}
Let $n,k,\ell\in\N$ with $k\ge8$, and $\ell\geq k^2\log n$.
We can construct a uniform $(n,k,\ell)$-splitter of size at most 
$k^2\log n/\log\ell$ in linear time.
\end{lemma}

\begin{proof}
Use functions $f_m\colon [n]\mapsto[m], x\mapsto x\bmod m$
for the last $r=\lceil k(k-1)\log n/(4\ln k+2\ln\log n-2\ln2))\rceil$
many prime numbers
smaller than $k^2\log n$. 
Finding every prime number smaller than $k^2\log n$
through the sieve of Erathostenes takes $O(k^2\log n \log(k^2\log n)$, which
is possible in linear time in the size of the splitter. 
If we choose the functions $f_m$ in this way
then the smallest $m$ is at least $k^2\log n/2$
because $\pi(k^2\log n)-\pi(k^2\log n/2)\geq r$, see \Cref{obs:primenumthm}.

The product of all these prime numbers is at least
$(k^2\log n/2)^r\geq n^{k(k-1)/2}$.
For bigger values of $\ell$ we can choose analogously the largest $r'=k^2\log n/\log \ell$ primes smaller than $\ell$
without reaching $\ell-k^2\log n/2$, and $(\ell-k^2\log n/2)^{k^2\log n/\log \ell}\ge n^{k(k-1)/2}$.

If $S\subseteq[n]$, $|S|=k$, then there are ${k\choose2}$ pairs
$x,y$ from~$S$.  If we take the difference $|x-y|$ of such a pair, it is
not zero.  The product of all differences of all pairs is a number $z$
smaller than $n^{k(k-1)/2}$, but strictly greater than zero.  By the
Chinese remainder theorem there is
an $m$ such that $f_m(z)\neq0$.  This implies that $f_m(x)\neq f_m(y)$
for \emph{every} pair~$x,y$.  For every possible $S$ we have at least
one injective function, so they form a splitter.  The splitter is
uniform because every modulo function is uniform.
\end{proof}

%



\begin{lemma}\label{cor:s2}
Let $n,k,\ell\in\N$ and $\ell\geq k^3$.
Then we can construct a 
$\lceil n/(2^{\ell/k^2}- k^2\log n/2)\rceil$-uniform $(n,k,\ell)$-splitter
size $O(k^4\log n/\log \ell)$ in linear time,
and if $k\ge\log\log n$ we can also construct a uniform $(n,k,\ell)$-splitter of
size $O(k^5\log n/\log \ell)$ also in linear time.
\end{lemma}

\begin{proof}
If $\ell\geq k^2\log n$, then  \Cref{lem:s1} yields the desired
result.  Therefore, we can assume that $ k\leq\log n$.

In this case, we can choose $\ell'\ge\lceil k^2\log n\rceil$ and apply  \Cref{lem:s1}.
This gives us a uniform $(n,k,\ell')$-splitter of size $O(k^2\log n/\log \ell')$.
We choose the $\ell'$ such that 
$\ell\ge k^2\log \ell'$ and we can use  \Cref{lem:s1} again to get
a uniform $(\ell',k,\ell)$-splitter of size~$O(k^2\log \ell'/\log \ell)$.  
Combining both splitters yields an $(n,k,\ell)$-splitter.  Its size is the
product of the two sizes, i.e., $O(k^4\log n/\log \ell)$.

In particular what we do is a bit more subtle than that. We know that the first splitter will contain functions
$f_1\colon[n]\to [m']$ for some $m'\le \ell'$, and for every particular $f_1$ 
we can construct a uniform $(m',k,\ell)$-splitter of size 
$O(k^2\log m'/\log\ell)$, which is even better than expected. The $(m',k,\ell)$-splitter
corresponding to each $f_1$ will have functions $f_2\colon[m']\to [m]$ for some $m\le \ell$.
Each function in the final splitter is the result of composing two modulo functions. 
For instance, $f=f_2\circ f_1$ yields a function that maps $[n]\to[m]$, 
now we ask ourselves what is the uniformity of $f$.

We know that the biggest preimage of $f_2$ has to be $f_2^{-1}(0)$ and let $r$ be an integer
such that $|f_2^{-1}(r)|=|f_2^{-1}(0)|-1$. The elements of
$f_2^{-1}(0)$ and $f_2^{-1}(r)$ are the preimage
of a modulo function, so they are evenly distributed at distance $m$ from each other.
Similarly biggest preimages of $f_1$ have to be concentrated in the smallest integers, let $i$ be the first value such that
$\lfloor n/m'\rfloor =f_1^{-1}(i)<f_1^{-1}(0)$.
Imagine that there are $s$ values of $f_2^{-1}(0)$ that are smaller than $i$, and $t$ that are larger,
correspondingly, there are $s'$ values of $f_2^{-1}(r)$ smaller than $i$, and $t'$ that are larger.
But we know that the preimage of $0$ and $r$ are alternated and only have a difference of one, 
so we know that $s+t=s'+t'-1$ and $s'\ge s-1$ and $t'\ge t-1$, thus, at worst $s'=s-1$ and $t'=t$,
and $|f^{-1}(0)|-|f^{-1}(r)|\le \lceil n/m'\rceil \le \lceil n/(\ell'-k^2\log n/2)\rceil$. This gives us
the uniformity of the combined splitter.

On the other hand, we can make this splitter uniform by using \Cref{lem:smooth},
where $a=\lceil n/(\ell'-k^2\log n/2)\rceil$, if $\ell'\ge \ell(k+1)+k^2\log n/2 $.

Recall that if $\ell\ge k^2\log n$ we can use 
\Cref{lem:s1} directly and obtain a uniform splitter of the desired size.
Last but not least, we need to appropriately choose $\ell'$. 
We want to make $\ell'$ as large as possible in order to have a splitter that is as uniform as possible,
and thus also easier to make completely uniform. So, we choose\
$\ell'=\lfloor 2^{\ell/k^2}\rfloor\le  n$ and obtain an $\lceil n/(2^{\ell/k^2}- k^2\log n/2)\rceil$-uniform $(n,k,\ell)$-splitter of size
$O(k^4\log n/\log \ell)$, and we can make it a 
uniform $(n,k,\ell)$-splitter of size $O(k^5\log n/\log \ell)$ with \Cref{lem:smooth} if $2^{\ell/k^2}\ge \ell(k+1)+k^2\log n/2$,
so, in particular if $k\ge \log \log n$ this is always possible.
\end{proof}

%

We could ask ourselves if we can reach even smaller values of $k$ by applying \Cref{lem:s1}
more times consecutively. But, every time we apply this lemma, we get an additional factor of
$k^2$. Moreover, the uniformity only gets worse with consecutive applications of the lemma. 
However, it is not necessary to do this, as once $k$ is small enough it is very easy to 
construct a splitter with brute force, as we see in the following lemma.

\begin{lemma}\label{lem:bruteforce}
Let $k\leq(\log\log n)^2$, and $\ell\ge k^2$.
We can construct a strongly uniform $(\lceil(\log\log n)^6\rceil,k,\ell)$-splitter of size
$O(k\log\log\log n)$ in linear time.
\end{lemma}

\begin{proof}
Let $t=\lceil(\log\log n)^6\rceil$.
Let $\F$ be the set of all functions $[t]\to[\ell]$ that are
balanced, i.e., if $f\in\F$ then $|f^{-1}(i)|\in
\{\lfloor t/\ell\rfloor, \lceil t/\ell\rceil\}$.
Let us first generously estimate, how big $\F$ is.  Clearly, $|\F|\leq
\ell^{t}\le t^t \leq(\log\log n)^{O((\log\log n)^6)}=O(n^{1/4})$.

If $f\in\F$ is chosen randomly then $\Pr[|f(S)|=k]\geq
\bigl({\ell-k\over \ell}\bigr)^k\geq \frac14$.  If $\S$ is an arbitrary family of
size-$k$ subsets of~$[t]$ then an expected number of a quarter
of its sets will be mapped injectively by a random $f\in\F$.
By trying all functions in $\F$ we can find such a function
deterministically in time $O(|\F|\cdot|\S|)=O(n^{1/2})$.

When we find such a function we reduce the number of sets that
are not yet mapped injectively by a factor of~$3/4$.  Hence,
we have to find only $O(k\log t)=O(k\log\log\log n)$ such functions.

This splitter will be strongly uniform by construction.
\end{proof}

\begin{theorem}\label{thm:k3splitter}
Let $k,\ell,n\in\N$ with $\ell\ge k^3$.  We can construct 
an $\lceil n/(\ell(2^{\ell/k^2}- k^2\log n/2))\rceil$-uniform $(n,k,\ell)$-splitter 
of size $O(k^4\log n)$ if $k\ge \log \log n$, 
an $\lceil n/k^2\rceil$-uniform $(n,k,\ell)$-splitter 
of size $O(k^5\log n)$ if $k\le \log \log n$, 
or for any value of $k$, a uniform $(n,k,\ell)$-splitter of size
$O(k^6\log n)$.
\end{theorem}

\begin{proof}
We look at different cases.  Firstly, if $k\geq\log\log n$ then
 \Cref{cor:s2} gives us an
 $\lceil n/(2^{\ell/k^2}- k^2\log n/2)\rceil$-uniform 
 $(n,k,\ell)$-splitter
of size $O(k^4\log n/\log \ell)$,
and a uniform $(n,k,\ell)$-splitter
size $O(k^5\log n/\log \ell)$ also in linear time.

%

Secondly, if $k\leq\log\log n$ then we can use  \Cref{cor:s2}
with $\ell'=\lceil(\log\log n)^6\rceil$ to construct an $\lceil
n/(2^{\ell'/k^2}- k^2\log n/2)\rceil$-uniform  $(n,k,\ell')$-splitter of
size $O(k^4\log n/\log\log\log n)$.  Then we use  \Cref{lem:bruteforce}
to get another strongly uniform $(\ell',k,\ell)$-splitter of size
$O(k\log\log\log n)$.  Combining them we have to see about the uniformity
of the splitter we create.  As we did in the proof of \Cref{cor:s2},
we have to take into account that for every possible image, we create a
different splitter using \Cref{lem:bruteforce}, however, the nonuniformity
in this case is a bit different because we do not combine two modulo
splitters, but two arbitrary ones. Thus, in the worst case the preimage
of two different elements can be as low as
$|f^{-1}(i)|=\lfloor \ell'/\ell\rfloor (n/\ell')$
and as high as
$|f^{-1}(j)|=\lceil \ell'/\ell\rceil (n/\ell'+\lceil
n/(2^{\ell'/k^2}- k^2\log n/2)\rceil)$,
which means that their difference
is at most
$\lceil n/\ell'+(\ell'/\ell+1)n/(2^{\ell'/k^2}- k^2\log n/2)\rceil$,
which, if we consider $\ell'\ge k^6$, $\ell'=\lceil(\log\log
n)^6\rceil$,
$\ell\ge k^3$ and $k\le \log\log n$,
in the appropriate places we can bound the
difference by
\begin{multline*}
\biggl\lceil \frac{n}{\ell'}+\Bigl(\frac{\ell'}{\ell}+1\Bigr)\frac{n}{2^{\ell'/k^2}-\frac12 k^2\log n}\biggr\rceil
\le
\frac{n}{k^6}+ \left(\frac{(\log\log n)^6+1}{k^3}+1\right)\frac{n}{2^{(\log \log n)^4}}+1\\
{}\le{} \frac{n}{k^6}+ (\log\log n)^6\frac{n}{2^{(\log \log n)^4}}
\leq\frac{2n}{(\log\log n)^6}
\leq\frac{2n}{k^6}\text{ for~$n\geq16$ and $k\geq2$.}
\end{multline*}
To see the correctness of the last inequality note that $x^{12}<2^{x^4}$ for $x\geq2$.
We can smooth this one out using \Cref{lem:smooth}, if $n\ge (k+1)\ell 2n/k^6$, which
follows from $\ell\le k^6/(2(k+1))$.  We can just choose $\ell=k^3$ to meet
this condition and note that any $a$-uniform $(n,k,k)^3$-splitter is
also an $a$-uniform $(n,k,\ell)$-splitter for $\ell>k^3$.
In this way we get an uniform $(n,k,\ell)$-splitter of size $O(k^6\log n)$.
\end{proof}

%


\section{Bisectors}\label{sec:bisect}

Let us begin by formally defining the notion of bisector.

\begin{definition}
Let $n$ and $k$ be integers, and $0\leq\alpha\leq1$.  A set $\F$ of
functions $[n]\to\{0,1\}$ is an \emph{$(n,k,\alpha)$-bisector} if
$|f^{-1}(1)|=\lceil\alpha n\rceil$ and for every $S\subseteq[n]$ with $|S|= k$, 
there is some $f\in\F$ such that 
$S\subseteq f^{-1}(0)$.
\end{definition}

Intuitively, in a bisector we choose some binary functions so that for every set $S$ of size $k$, there
is one function that maps all of its elements to $0$.
Just as in the case of splitters, constructing a bisector consists on choosing one such set of functions and listing 
them as a table. This means that the time complexity of building a bisector is $\Omega(n|\F|)$.

The goal is to construct small bisectors quickly (in linear time). One idea that comes to mind is to
start with a function that maps every value to 0 (or to 1) and in consecutive steps
choose subsets of elements of $[n]$ and map those to 1 (or to 0), as is depicted on 
the left-hand (right-hand) diagram of \cref{fig:bisect}, while taking care that the
elements in our $k$-sets end up in the desired category after each step. In the case
of bisectors we will use the first strategy, but a mix of the first and second 
strategies will be necessary to build universal sets as we will see in \Cref{sec:usets}.

\begin{figure}
\centerline{\includegraphics{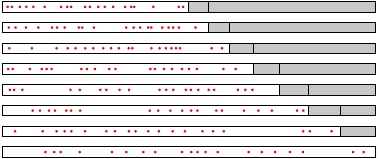}\hfil
\includegraphics{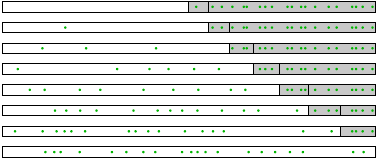}}
\caption{Two strategies to construct a bisector. In the left picture we start with a function that maps all elements to $0$ and on
each step we map some elements to $1$ leaving out a potential $k$-set. In the right picture, 
we start with a function that maps all elements to $1$ and in every step we make sure that a significat part of a given a set
$S$ is mapped to $0$.}
\label{fig:bisect}
\end{figure}

We begin with a simple case, where the number of elements mapped to 1 is small in the bisector.
This corresponds to only one step in \cref{fig:bisect}.

\begin{lemma}\label{lem:bis1}
For every constant $c>0$, there exists a $k'$ such that for every integer $k\ge k'$  
we can construct a $(ck^3,k,1/\sqrt k)$-bisector of size at most 
$6ke^{\sqrt k}\ln c k$ in time
$ k^{O(k^{5/2})}$.
\end{lemma}

\begin{proof}
Let $=S\subseteq[ck^3]$, $|S|=k$.  There are at most $k^{ck^3}$ ways to
choose such a set~$S$.  A function $f\colon[ck^3]\to\{0,1\}$ from an
$(ck^3,k,1/\sqrt k)$-bisector has to split $[ck^3]$ into two parts of
sizes $\lceil c k^{5/2}\rceil$ 
and $ck^3-\lceil c k^{5/2} \rceil$.  
There are only ${c k^3\choose \lceil c k^{5/2}\rceil}
\leq (c k^3)^{c k^{5/2}}$
possibilities for such a function and we can construct a
first $(c k^3,k,k^{-1/2})$-bisector $\F$ that consists of \emph{all} such
functions.  Of course, this family is still too big.

If $S$ is a fixed, but arbitrary subset of $[c k^3]$ with size~$k$ then
\[
\Pr[S\subseteq f^{-1}(0)]=\prod_{i=0}^{k-1}\frac{c k^3-\lceil c k^{5/2}\rceil -i}{c k^3}
\geq(1-k^{-1/2}-k^{-2}/c)^k\geq\frac12e^{-\sqrt k}
\]
if $f$ is chosen randomly from~$\F$. For the last step, for every value of $c> 0$,
there exists a $k'$ such that the last step is true 
for every $k\ge k'$. One can easily check that the difference grows 
monotonically for larger values of $k$. Moreover, for $c\ge 1$ it suffices to take $k\ge 16$.

Let now $\S$ be a family of sets~$S\subseteq[ck^3]$ of size $k$.  Then
\[
E\Bigl(|\{\,S\in\S\mid S\subseteq f^{-1}(0)\,\}|\Bigr)\geq \frac12e^{-\sqrt
k}|\S|.
\]
Hence, there is some $f\in\F$ that splits at least $e^{-\sqrt
k}|\S|/2$ sets from $\S$ in the right way.

Our strategy is to choose such an $f$ and add it to new family~$\F'$
which now splits more sets in~$\S$ in the correct way.  Most sets
might remain unsplit, but we can then look for a good $f$ for
\emph{them}.  Repeating this idea over and over leads to two sequences
$\S=\S_1,\S_2,\S_3,\ldots$ and $\emptyset=\F_1,\F_2,\F_3,\ldots$ where
$\S_i$ contains all sets from $\S$ that are not yet split correctly
using functions from~$\F_i$.  We can construct $\F_{i+1}$ from $\F_i$
by finding one $f\in\F$ and setting $\F_{i+1}=\F_i\cup\{f\}$.  This
means in particular that $\F_i$ contains exactly $i-1$ many functions.

Moreover, $|\S_{i+1}|\leq(1-e^{-\sqrt k}/2)|\S_i|$ and we have to
analyze for which $t$ the set $\S_t$ is guaranteed to be empty,
meaning that for every $S\in\S$ there if a function in $\F_t$ that
splits that $S$ correctly.

It is easy to see that $\S_t=\emptyset$ if $(1-e^{-\sqrt k})^t|\S|<1$,
which is true when $t>2e^{\sqrt k}\ln|\S|$.
Initially $\S$ contains all possible set $S$ of size~$k$, which means
$|\S|\leq c k^{3k}$ and therefore $t>6ke^{\sqrt k}\ln (ck)$ is sufficient
to construct $\F'=\F_t$.

How long does it take to compute $\F_{i+1}$
from $\F_i$ and $\S_i$?  We have to test every $f\in\F$ on every
$S\in\S_i$.  One such test can be carried out in $O(k)$ steps,
and then we need to list each~$f$, which takes $O(n)$ time (in this particular case $n=ck^3$).  The
total time is then $O(n\cdot k\cdot|\F|\cdot|\S|)=c k^3 \cdot k^{O(k^{5/2})}=k^{O(k^{5/2})}$.
\end{proof}

The following lemma will allow us to extend bisectors by a few elements if we have a smaller bisector already.
The construction is very similar to the one in Lemma 4 in~\cite{DreierLR2020}.

\begin{lemma}\label{lem:notexact}
If an $(n,k,\alpha)$-bisector $\F$ is given with $d(k+1)<n$ and $d\in\bf N$,
then we can construct in linear time an 
$(n+d,k,\alpha)$-bisector of size $(k+1)|\F|$.
\end{lemma}

\begin{proof}
In order to do this we define a family of functions
$[n+d]\to\{0,1\}$ that maps all sets $S\subseteq[n+d]$ of size $k$
correctly under the condition that $S$ avoids $d$ given elements
of~$[n+d]$.  We can choose $k+1$ such forbidden sets that are mutually
exclusive and construct the union of the respective families.
We do this as follows, we choose $k+1$ mutually exclusive sets of size 
$d$ in $[n]$, for instance the first set could contain $[d]$ then 
the second set from $d+1$ to $2d$ etc. We can do this because we know that
$d(k+1)<n$. Then we take the bisector $(n,k,\alpha)$ and apply it to $[n+d]$
ignoring one of the forbidden sets every time. We still have to fully extend the functions,
but this can be done arbitrarily in any way that still preserves the desired proportion of $0$ and
$1$ in the functions. The number of functions is now $(k+1)|\F|$, as for each forbidden set we have a 
full bisector that then we extend.

We argue that this is a valid $(n+d,k,\alpha)$-bisector,
as for any set $S$ of size $k$ we can find at least one of the $k+1$ sets of size $d$ containing 
no elements from $S$, the bisector constructed through that forbidden set will contain a function
mapping all elements of $S$ to $0$.
This
gives us a family of the desired size, with the time complexity only
increased by a factor of $k+1$.
\end{proof}

%

\begin{lemma}\label{lem:extendn0}
 Let $n_1$, $n_2$, and $k$ be integers with $k\le n_2\le n_1$, and
 let $0\le \alpha \le 1$. If one can construct
 $(n_2,k,\alpha)$-bisector $\F$
 of size $|\F|$ in time $t$, 
 one can construct an $(n_1,k,\alpha)$-bisector of size $k|\F|$
 in time $O(t + k|\F|n_1)$.
\end{lemma}

\begin{proof}
 Let $c$ and $d$ be the positive constants such that $n_1=cn_2+d$ with $d<n_2$.
 Let us consider for $m=cn_2$ the modulo function
 $mod_{n_2}:[m]\to [n_2]$ that takes
 every number and maps it to its modulo with respect to $n_2$.
 For any given set $S$ of size $k$ in $[m]$, one can consider the
 subset $S'$ that one obtains when mapping all of its elements with
 the modulo function. This set might not be of size $k$, but one can
 complete it to a set $S''$ of size $k$ by adding arbitrary elements.
 Then, one can construct an $(m,k,\alpha)$-bisector 
 by taking modulo $n_2$ of each function of an $(n_2,k,\alpha)$-bisector.
 If the value of $n_1$ is multiple
 of $n_2$ one obtains
 an $(n_1,k,\alpha)$-bisector of the same size as the one given 
 in time $c't$ for some constant $c'$. 
 Otherwise, obtain a bisector for $m$ as described and then apply \Cref{lem:notexact}, with $d<n_2$
 to obtain a bisector of size $k|\F|$.
 We use time $t$ to construct the first bisector and everything else can be done
 in linear time.
\end{proof}

\begin{corollary}\label{cor:bisn1}
Let $k\geq16$ be an integer and $n\ge k^4$.  
We can construct an $(n,k,1/\sqrt k)$-bisector 
of size at most $6k^2e^{\sqrt k}\ln k$ in time
$nk^{O(k^{5/2})}$.
\end{corollary}

\begin{proof}
 One can apply \Cref{lem:extendn0} to the bisector obtained in
 \Cref{lem:bis1} when $c=1$. 
\end{proof}

\begin{lemma} \label{lem:bisn1}
Let $n\geq k^4$ and $k\ge16$ be an integer, and let $0\le \alpha< 1$ be a constant.  
We can construct an 
$(n,k,\alpha)$-bisector of size $\bigr(\frac{1}{1-\alpha}\bigl)^kk^{O(\sqrt k)}$ in $nk^{O(k^{5/2})}$ time.
\end{lemma}

\begin{proof}
First take an $(n,k,1/\sqrt k)$-bisector as given by \Cref{cor:bisn1}.  
The functions in this
bisector map exactly $n'=\lfloor n-n/\sqrt k\rfloor$ numbers from $[n]$ to~$0$.
We can then apply, to the numbers mapped to 0 of each of the functions of the previous 
bisector another function from an $(n',k,1/\sqrt
k)$-bisector.  This function maps $\lfloor n'-n'/\sqrt k\rfloor$ numbers
to~$0$.  We can repeat this until only $(1-\alpha)n$ numbers are mapped
to~$0$.  How many repetitions do we need?  In each iteration the
number of $0$s shrinks by a factor of $1-1/\sqrt k$ and we need
$(1-1/\sqrt k)^t\leq (1-\alpha)$, which means, due to $-\ln (1-x)\ge x$ for $x\le 1$, that
$t=\lceil\sqrt k\ln\frac{1}{1-\alpha}\rceil$ iterations are sufficient for
the desired fraction of 0s. One needs to take into account that on the last step one
might need to adjust the last constant, to do this one can take a smaller fraction of the 
previous number of 0s to form the last bisector.  Each of the bisectors
that have to be combined has size at most $6k^2e^{\sqrt k}$ so their
combination has size at most $(6k^2e^{\sqrt k})^t=\bigr(\frac{1}{1-\alpha}\bigl)^kk^{O(\sqrt k)}$.

The bisector of each iteration takes $nk^{O(k^{5/2})}$ time to build, there are $O(\sqrt{k})$
such iterations, so the total time is still $nk^{O(k^{5/2})}$, and then composing all of them
together takes only linear time.
\end{proof}

\begin{lemma}\label{lem:gen0small}
Let $k$ be an integer, $0\le\alpha<1$ be constant, and $n=k^{O(1)}$, but $n\geq k^5$.  We can construct
an $(n,k,\alpha)$-bisector of size $\bigr(\frac{1}{1-\alpha}\bigl)^kk^{O(k^{5/6})}$ in linear time.
\end{lemma}

\begin{proof}
For a set $S\subseteq[n]$ with $|S|=k$ we can guess $\ell=\lceil
k^{2/3}\rceil$ or $\ell=\lfloor k^{2/3}\rfloor$ disjoint
intervals $I_1,\ldots I_\ell$ such that each interval contains
either $\lceil k^{1/3} \rceil$ or $\lfloor k^{1/3} \rfloor$
elements from~$S$.  We can also guess one interval of size $k^4$ that
does not contain any element from~$S$.  We build $\ell$ sets $I'_i$ by
combining $I_i$ with $k^3$ elements from the big interval.  In that way, for all possible guesses,
we have sets $I'_i$ with the following two important properties:  $I'_i$
contains exactly $\lceil k^{1/3}\rceil$ (or $\lfloor k^{1/3}
\rfloor$) elements from~$S$ and the size of $I'_i$ is at least~$k^3$.
We can construct, using \Cref{lem:bisn1}, for each set $I'_i$ with
size $n'_i$, an $(n'_i,\lceil k^{\frac13}\rceil,\alpha)$-bisector
(an $(n'_i,\lfloor k^{1/3}\rfloor,\alpha)$-bisector resp.) of size
$\bigr(\frac{1}{1-\alpha}\bigl)^{\lceil k^{1/3}\rceil}k^{O(k^{1/6})}$
($\bigr(\frac{1}{1-\alpha}\bigl)^{\lfloor
k^{1/3}\rfloor}k^{O(k^{1/6})}$ resp.), in time $nk^{\frac13
O((k^{1/3})^{5/2})}=k^{O(k^{5/6})}$.  Combining all
$k^{2/3}$ families into a big one yields a family of size
$\bigr(\frac{1}{1-\alpha}\bigl)^kk^{O(k^{5/6})}$.
\end{proof}

But this only works if the size of the bisector is polynomial in $k$, now
we need a result similar to \Cref{cor:bisn1} but for general $\alpha$.

\begin{theorem}\label{thm:gen0large}
Let $k\geq16$ be an integer and $n\ge k^4$.  
We can construct an $(n,k,\alpha)$-bisector 
of size at most $\bigr(\frac{1}{1-\alpha}\bigl)^kk^{O(k^{5/6})}$ in 
linear time.
\end{theorem}

\begin{proof}
 One can trivially apply \Cref{lem:extendn0} to a bisector obtained in
 \Cref{lem:gen0small}. 
\end{proof}

\section{Uniform Universal Sets}\label{sec:usets}

Universal sets are families of functions from $[n]$ to $\{0,1\}$
that map ever possible $k$-subset of~$[n]$ to every possible combination
of $0$s and $1$s. In particular, for every $k$-subset there must
be a function in the universal set that maps every element to $0$, so a universal set
could also be a bisector if the functions behave appropriately.

In order to do this, we define a special class of universal sets. We want them to be also bisectors,
but not only that, we want to be able to choose the fraction of $0$s and $1$s.
So we define \emph{$(n,k,\alpha)$ universal sets} in a way analogous to the bisectors.

\begin{definition}
Let $n$, $k$ be integers and $0<\alpha<1$.  A set $\F$ of
functions $[n]\to\{0,1\}$ is an \emph{$(n,k,\alpha)$-universal
set} if for every $f\in \F$, 
$|f^{-1}(1)|=\lceil\alpha n\rceil$,
and for every
$S_0,S_1\subseteq[n]$ with
$|S_0|+|S_1|=k$ and $S_0\cap S_1=\emptyset$
there is some $f\in\F$ such that $S_0\subseteq f^{-1}(0)$
and $S_1\subseteq f^{-1}(1)$.
\end{definition}

In order to adapt the results obtained for \Cref{sec:bisect} to work for universal sets we will need
to build some special families of functions that we name \emph{mapping families}.

\begin{definition}
Let $n$, $k_0$, $k_1$ be integers and $0\leq\alpha,\beta\leq1$.  A set $\F$ of
functions $[n]\to\{0,1\}$ is an \emph{$(n,k_0, k_1,\alpha,\beta)$\mapfam} if
$|f^{-1}(1)|=\lceil\alpha n\rceil$ and for every $S_0\cup S_1 = S\subseteq[n]$ with
$S_0\cap S_1=\emptyset$, $|S_0|= k_0$, $|S_1|=  k_1$
there is some $f\in\F$ such that 
$S_0\subseteq f^{-1}(0)$ and $|S_1\cap f^{-1}(1)|=\lceil\beta k_1\rceil$.
\end{definition}

We now generalize the previous results, the goal is to obtain
a result like the one in \Cref{thm:gen0large} but for mapping families.
We start by generalizing \Cref{lem:bis1}.

\begin{lemma}\label{lem:bis2}
Let $k_0, k_1$ be such that $k_0+k_1\le k$ and let $0\le \beta \le 1$.
For every $c> 0$ there exists a $k'$ such that for every $k\ge k'$,   
we can construct a $(ck^3,k_0,k_1,1/\sqrt k,\beta)$\mapfam of size at most 
\[ 
 8{k_1 \choose \lceil\beta k_1\rceil}^{-1}\bigl(\sqrt{k}\bigr)^{\lceil\beta k_1\rceil}
 e^{\frac{k_0+k_1}{\sqrt{k}}}k\ln ck 
\]
 in time
$nk^{O(k^{5/2})}$.
\end{lemma}

\begin{proof}
Let $S_0\cup S_1= S\subseteq[ck^3]$, $S_0\cap S_1=\emptyset$, $|S_0|= k_0$,
$|S_1|= k_1$. One can take the same approach as in the proof of \Cref{lem:bis1}.
The number of possible sets $S_0$ and $S_1$ can be bounded by 
${ck^{3} \choose k_0+k_1} {k_0+k_1 \choose k_1}\leq ck^{3(k_0+k_1)}(k_0+k_1)^{k_1}\le ck^{4k} $,
and the number of candidate functions is ${ck^3\choose \lceil ck^{5/2}\rceil }\leq ck^{3ck^{5/2}}$. 
Let $\F$ be a $(ck^3,k_0,k_1,k^{-1/2},\beta)$\mapfam consisting
of all such functions.

One needs to do a selection of functions of $\F$ in such a way that on every step 
the newly added function covers a large portion of the remaining sets, 
as we did for \Cref{lem:bis1}.
If $S_0$ and $S_1$ are fixed, but arbitrary disjoint subsets of $[ck^3]$ with size 
$k_0$ and $k_1$ respectively, then
\begin{multline*}
\Pr[S_0\subseteq f^{-1}(0)\wedge |S_1\cap f^{-1}(1)|=\lceil\beta k_1\rceil]\\
={k_1 \choose \lceil\beta k_1\rceil}
\left(\prod_{i=0}^{\lceil\beta k_1\rceil-1}\frac{\lceil ck^{5/2}\rceil -i}{ck^3}\right)
\left(\prod_{i=0}^{k_0+k_1-\lceil\beta k_1\rceil-1}\frac{ck^3-\lceil ck^{5/2}\rceil-i}{ck^3}\right) \\
\ge {k_1 \choose \lceil\beta k_1\rceil}(k^{-1/2}-k^{-2}/c)^{\lceil\beta k_1\rceil}
(1-k^{-1/2}-k^{-2}/c)^{ k_0+k_1-\lceil\beta k_1\rceil}\\
={k_1 \choose \lceil\beta k_1\rceil}(\sqrt{k})^{-\lceil\beta k_1\rceil}
(1-k^{-\frac32}/c)^{\lceil\beta k_1\rceil}(1-k^{-1/2}-k^{-2}/c)^{k_0+k_1-\lceil\beta k_1\rceil}  \\
\geq{k_1 \choose \lceil\beta k_1\rceil}( \sqrt{k})^{-\lceil \beta k_1\rceil}
(1-k^{-1/2}-k^{-2}/c)^{k_0+k_1}
\ge\frac12{k_1 \choose \lceil\beta k_1\rceil}( \sqrt{k})^{-\lceil \beta k_1\rceil}e^{-\frac{k_0+k_1}{\sqrt{k}}}\;,
\end{multline*}
if $f$ is chosen randomly from~$\F$, where, in the last step we used the
same inequality as in the proof of \Cref{lem:bis1}, that is, for every
$c>0$ there exists a $k'$ such that for every $k\ge k'$ the inequality
holds.

From here one can, given a family $\S$ of pairs~$S_0$ and $S_1$,
compute the expected number of sets covered by one function
\[
E\Bigl(|\{\,S\in\S\mid S_0\subseteq f^{-1}(0)\wedge|S_1\cap f^{-1}(1)|=\beta k_1\,\}|\Bigr)
\geq \frac12{k_1 \choose \lceil\beta k_1\rceil}( \sqrt{k})^{-\lceil \beta k_1\rceil}e^{-\frac{k_0+k_1}{\sqrt{k}}}|\S|\;.
\]
and select a function that covers
more than that in each step.
So we can define a series of families where $\S_i$ contains the pairs of sets that are not covered by the 
family of selected functions $\F_i$. One then needs to compute the $t$ such that $\S_t=\emptyset$, and
considering that $|\S|\leq ck^{4k}$ and therefore
\[
 t>8{k_1 \choose \lceil\beta k_1\rceil}^{-1}( \sqrt{k})^{\lceil \beta k_1\rceil}e^{\frac{k_0+k_1}{\sqrt{k}}}k\ln ck\;.
\]
As in \Cref{lem:bis1}, the time to construct such a mapping family is $O(n\cdot k\cdot|\F|\cdot|\S|)=n k^{O(k^{5/2})}$.
\end{proof}

Now, to generalize the result in  \Cref{lem:bis2} 
or even the result in  \Cref{cor:bisn1} to any value of $n$ we need a result equivalent
to  \Cref{lem:extendn0} but with general values for $k_1$ and $\beta$. We cannot do it
by directly using modulo functions, as elements of $S_0$ and $S_1$ would clash. To avoid this
we could use Lemma~2 from Naor et al.\ \cite{NaorSS95}, which uses good error correcting codes. 
However, this would not give us mapping families as the splitters they construct are not uniform.
This is not a problem with the uniform splitters of the same size that we build in 
\Cref{sec:splitters}. This
will result in a small blowup in the size of
the mapping family and the time of its construction.

\begin{lemma}\label{lem:extendn}
 
 For every $n$ there exists a $k'$
 such that for every integer $k\ge k'$,
 given $k_0$, and $k_1$ be integers with $k\ge k_0+k_1$ and $n\ge k^4$.
 We can construct an $(n,k_0,k_1,1/\sqrt{k},\beta)$\mapfam of size 
   \[{k_1 \choose \lceil\beta k_1\rceil}^{-1}( \sqrt{k})^{\lceil \beta k_1\rceil}
  e^{\frac{k_0+k_1}{\sqrt{k}}}O(k^8\log k\log n)\]
 in time $ nk^{O(k^{5/2})}$
\end{lemma}

\begin{proof}
 Let us consider the uniform $(n,k,k^3)$-splitter of size $O(k^6 \ln n)$ constructible in linear time
 given in \Cref{thm:k3splitter}. Each function in this splitter maps $[n]$ to $[ck^3]$ uniformly for some
 $c>0$.
 
 Let us take the $(ck^3,k_0,k_1,1/\sqrt{k},\beta)$\mapfam given by  \Cref{lem:bis2}. We can 
 construct an $(n,k_0,k_1,1/\sqrt{k},\beta)$\mapfam composing the functions of the splitter
 from \Cref{thm:k3splitter} with the functions from the mapping family in  \Cref{lem:bis2}.
 We have to count now, how many elements from $[n]$ are mapped to $0$ and $1$
 respectively. Because the $(n,k,k^3)$-splitter is uniform, a function 
 $f\colon[n]\to[ck^3]$ of this splitter will have an 
 $|f^{-1}(0)|=\lceil n/ck^3\rceil$,
 or $|f^{-1}(0)|=\lfloor n/ck^3\rfloor$ so, when combining this splitter function
 with its corresponding mapping family, we will get a function $g\colon[n]\to\{0,1\}$
 with
 $\lceil n/ck^3\rceil\lceil ck^3/\sqrt{k}\rceil\ge |g^{-1}(1)|\ge 
 \lfloor n/ck^3\rfloor\lceil ck^3/\sqrt{k}\rceil$.
 We want to have $|g^{-1}(1)|=\lceil n/\sqrt{k}\rceil$,
 however, the deviation from that is at most $O(k^3/\sqrt{k})$, because $n\ge k^4$,
 the deviation will be at most $O(n/k)$, we want to flip some of the $0$ to $1$ or
 viceversa to compensate for the deviation. 
 Let $S\subseteq[n]$ be a $k$-subset of $[n]$.
 We know that $|g^{-1}(1)|=\Omega(n/\sqrt{k})$ ,
 and  we need to flip $O(n/k)$ many, we want to do this without hitting $S$.
 We can divide $g^{-1}(1)$ into $k+1$ subsets of size $\Omega(n/k^{3/2})$,
 for every one of these subsets we will construct a new function
 $g'\colon[n]\to\{0,1\}$ that maps
 one of these subsets to $0$, there is always at least one subset which does not 
 contain an element of $S$.
 
 The size will be the product of their sizes times $k$ 
 to compensate for the unevenness
  \[{k_1 \choose \lceil\beta k_1\rceil}^{-1}( \sqrt{k})^{\lceil \beta k_1\rceil}
  e^{\frac{k_0+k_1}{\sqrt{k}}} k\ln k\cdot O(k^6\log n)\cdot k\;.\]
  And the construction time is the time to construct both and compose them, but this last step
  only takes linear time.
\end{proof}

\begin{lemma}\label{lem:bisn2} Let $k_0$ and $k_1$ be integers
let $k=k_0+k_1$, and let $n\ge k^4$.
We can construct an $(n,k_0,k_1,\alpha,1)$\mapfam of size $ \bigl(\frac{1}{1-\alpha}\bigr)^{k_0}\bigl(\frac{1}{\alpha}\bigr)^{k_1}(k\ln n)^{O(\sqrt{k})}$ in time $n k^{O(k^{5/2})}$ . 
\end{lemma}
This is an extension of  \Cref{lem:extendn}. Here, we do the equivalent step 
to what  \Cref{lem:bisn1} does for  \Cref{lem:bis1} but with a general value of $k_1$.
The proof details are very technical.
Now we can use a similar strategy as in  \Cref{lem:gen0small} to obtain 
mapping families for small values of $n$.

\begin{proof}
 Let us consider an $(n,k_0,k_1,1/\sqrt{k},\beta_1)$-mapping family, for some given $\beta_1$.
 We can construct such a mapping family using  \Cref{lem:extendn}. The functions
 in this mapping family map $n_2=n-\lceil n/\sqrt{k}\rceil$ numbers from $[n]$ to 0.
 Moreover, they map $k_2=k_1-\lceil\beta_1 k_1\rceil$ numbers of $S_1$ to 0
 and $\lceil \beta_1 k_1\rceil$ numbers of $S_1$ to 1.
 
 We can then consider an $(n_2, k_0, k_2, 1/\sqrt{k}, \beta_2)$-mapping family,
 with $k_0+k_2\le k$, which we can construct according to  \Cref{lem:bis2} and 
  \Cref{lem:extendn}.
 This mapping family maps $n_3=n_2-\lceil n_2/\sqrt{k}\rceil$ numbers of $n_2$ to 0, and also subdivides $k_1$ further
 into $k_3=k_1-\lceil\beta_1 k_1\rceil -\lceil\beta_2 k_2\rceil $ values that are still mapped to 0 and $\lceil\beta_2 k_2\rceil$
 that get mapped to 1 in the second step.
 
 One can concatenate these two mapping families by applying to each
 function $f$ of the first mapping family a function $g$ 
 of the second one on those values mapped to 0. To obtain an 
 $(n,k_0,k_1,\frac{2}{\sqrt{k}}-\frac1k, \beta_1+\beta_2(1-\beta_1))$-mapping family, whose size is the product of the sizes
 of the first and second mapping families.
 
 We can repeat this, adjusting the values of $n_i$, $\beta_i$ and $k_i$ accordingly until $\lceil\alpha n\rceil$
 numbers, and, in particular, all of $k_1$ is also mapped to 1.
 
 How many iterations do we need to map all of the necessary elements to 1? As we did in  \Cref{lem:bisn1}, to
 adjust the number of 0s, $t=\lceil\sqrt{k}\ln{\frac{1}{1-\alpha}}\rceil$ iterations are sufficient. Thus,
 we have to select $t$ values of $\beta$ and then multiply the sizes
 of all of those mapping families
 to find out the total size of the targeted mapping family, for all of the required elements to be mapped to 1 we just need
 to make sure that $\beta_t=1$.
 
 How do we select appropriate values of $\beta_i$ for $i=1\hdots t$?
 We already have fixed that $\beta_t=1$, as in the end we need to have all
 of $S_1$ mapped to 1. But in order to find the appropriate values of $\beta$ we have to take a look at
 the product of the mapping family sizes. If we take on each step a
 mapping family using  \Cref{lem:extendn}, we
 have a total size, after $\lceil\sqrt{k}\ln{\frac{1}{1-\alpha}} \rceil$ steps of 
 \[
 \prod_{i=1}^{\lceil\sqrt{k}\ln{\frac{1}{1-\alpha}} \rceil} 8{k_i\choose \lceil \beta_ik_i \rceil}^{-1}(\sqrt{k})^{\lceil \beta_i k_i \rceil}e^{\frac{k_0+k_i}{\sqrt{k}}}k\ln k\ln n\;.
 \]
 If we expand the binomial terms into their corresponding factorials, and take out of the product all of the terms that don't depend on $i$ we obtain
 \begin{align*}
 \frac{(8k\ln k \ln n)^{\lceil\sqrt{k}\ln{\frac{1}{1-\alpha}} \rceil-1}}{k_1!}\prod_{i=1}^{\lceil\sqrt{k}\ln{\frac{1}{1-\alpha}} \rceil}(\lceil \beta_i k_i \rceil)!(\sqrt{k})^{\lceil \beta_i k_i \rceil}e^{\frac{k_0+k_i}{\sqrt{k}}}\;.
 \end{align*}
 A way to find appropriate values for $\beta_i$ is to try to minimize the product, but for that we have too many variables. What we do instead
 is to consider only two factors of it, the one for $k_{i-1}$ and the one for $k_i$. 
 Given a fixed value $k_{i-1}$ and $k_{i+1}$, one can determine the value $k_{i}$ that minimizes the product.
 For that, we rewrite $\lceil\beta_{i-1} k_{i-1}\rceil= k_{i-1} - k_{i}$, and  $\lceil \beta_i k_i \rceil= k_i - k_{i+1}$, and we obtain
  \[
 (k_{i-1}-k_{i})!(\sqrt{k})^{k_{i-1}-k_{i}} e^{\frac{k_0+k_{i-1}}{\sqrt{k}}}
 (k_{i}-k_{i+1})!(\sqrt{k})^{k_{i}-k_{i+1}} e^{\frac{k_0+k_i}{\sqrt{k}}}\;,
 \]
 which can be regrouped as
 \[
 (k_{i-1}-k_{i})!(k_{i}-k_{i+1})!
 (\sqrt{k})^{k_{i-1}-k_{i+1}} e^{\frac{2k_0+k_i+k_{i-1}}{\sqrt{k}}}\;,
 \]
 to find the $k_i$ that minimizes this function, we can ignore all terms that do not depend on $k_i$,
 we can also take the logarithm of it and take 
 Stirling's factorial approximation, i.e., 
 for any integer $m$, $m\ln m -m \le \ln m! \le (m+1)\ln m - m+1\le (m+1)\ln (m+1) - m$.
 Then the logarithm of our product becomes,
  \begin{align*}
  (k_{i-1}-k_{i}+1)\ln(k_{i-1}-k_{i}+1)-(k_{i-1}-k_{i}) +\\
  (k_{i}-k_{i+1}+1)\ln(k_{i}-k_{i+1}+1)-(k_{i}-k_{i+1}) +
 \Bigl({\frac{k_i}{\sqrt{k}}}\Bigr)\;.
 \end{align*}
 And taking the derivative with respect to $k_i$ and setting it to $0$ should give us the minimal value
 \begin{align*}
 -\ln (k_{i-1}-k_{i}+1)&-1+1
 +\ln (k_{i}-k_{i+1})+1-1+\frac{1}{\sqrt{k}}=0\\
 &\ln\Bigl(\frac{ \beta_{i-1} k_{i-1} +1}{ \beta_i k_i +1}\Bigr)=\frac{1}{\sqrt{k}}\\
 &\beta_{i-1} k_{i-1}=( \beta_i k_i +1)e^{\frac{1}{\sqrt{k}}}-1\;.
 \end{align*}
 For simplicity we can take $\beta_{i-1} k_{i-1} =  \beta_i k_i  e^{\frac{1}{\sqrt{k}}}$,
 and then take the ceiling of that when necessary.
 To obtain the values of all $\lceil \beta_i k_i \rceil$, one can then apply this inequality recursively and get
 $\beta_{t-i} k_{t-i}= \beta_t k_t e^{\frac{i}{\sqrt{k}}}$, and we also know that
 $\sum_i \lceil \beta_i k_i \rceil =k_1$ and $\beta_t=1$. Thus,
 \[
 k_1=\sum_{i=1}^t \lceil \beta_i k_i \rceil  
 \ge \sum_{i=0}^{t-1} k_t e^{\frac{i}{\sqrt{k}}}
 =\frac{e^{\frac{t}{\sqrt{k}}}-1}{e^{\frac{1}{\sqrt{k}}}-1}k_t\;,
 \]
 which means that 
 \[
  \beta_i k_i  = \frac{e^{\frac{1}{\sqrt{k}}}-1}{e^{\frac{t}{\sqrt{k}}}-1}k_1 e^{\frac{t-i}{\sqrt{k}}}
 =\frac{e^{\frac{1}{\sqrt{k}}}-1}{e^{\frac{i}{\sqrt{k}}}(1-e^{\frac{-t}{\sqrt{k}}})}k_1
 =\frac{e^{\frac{1}{\sqrt{k}}}-1}{\alpha e^{\frac{i}{\sqrt{k}}}}k_1\;.
 \]
 We can then take  into consideration that $e^{\frac{1}{k}}-1=\frac{1}{\sqrt{k}}\bigl(1+O \bigl(\frac{1}{\sqrt{k}}\bigr)\bigr)$, 
 so
 \[\lceil \beta_i k_i \rceil = \frac{k_1}{\alpha\sqrt{k}}e^{-\frac{i}{\sqrt{k}}}\Bigl(1+O \Bigl(\frac{1}{\sqrt{k}}\Bigr)\Bigr)\;.\]
 The given value of $\beta_i k_i$ is a lower bound for the ceiling function, an upper bound can be easily achieved by adding 1. The values of $k_i$ can be upper bounded using  \Cref{obs:useful}:
 \begin{align*}
  k_i=& k_{i-1}-\lceil\beta_{i-1}k_{i-1}\rceil=k_1-\sum_{j=1}^{i-1}\lceil\beta_{j}k_j\rceil\\
  \le& k_1 -\sum_{j=1}^{i-1}\frac{e^{\frac{1}{\sqrt{k}}}-1}{\alpha e^{\frac{i}{\sqrt{k}}}}k_1
  = k_1\Bigl(1-\Bigl(\frac{e^{\frac{1}{\sqrt{k}}}-1}{\alpha }\Bigr)\sum_{j=1}^{i-1} e^{\frac{-i}{\sqrt{k}}}\Bigr)\\
  =&k_1\Bigl(1-\frac{1-e^{\frac{-i}{\sqrt{k}}}}{\alpha}\Bigr)\;.
 \end{align*}
 Now that we have candidates for the values of $k_i$ and $\lceil \beta_i k_i \rceil $, we can take a closer look at the product we had before and approximate some of the factors to get
\begin{align*} 
 &\frac{(8k\ln k \ln n)^{\lceil\sqrt{k}\ln{\frac{1}{1-\alpha}} \rceil-1}}{k_1!}\prod_{i=1}^{\lceil\sqrt{k}\ln{\frac{1}{1-\alpha}} \rceil}(\lceil \beta_i k_i \rceil )!
 (\sqrt{k})^{\lceil \beta_i k_i \rceil }e^{\frac{k_0+k_i}{\sqrt{k}}}\\
 \le{} &\frac{(k\ln n)^{O(\sqrt{k})}}{k_1!}\prod_{i=1}^{\lceil\sqrt{k}\ln{\frac{1}{1-\alpha}} \rceil}
 (\lceil \beta_i k_i \rceil )^{\lceil \beta_i k_i \rceil +1}e^{-\lceil \beta_i k_i \rceil  +1}
 (\sqrt{k})^{\lceil \beta_i k_i \rceil }e^{\frac{k_0+k_i}{\sqrt{k}}}\\
 ={}&\frac{(k\ln n)^{O(\sqrt{k})}}{k_1!}\prod_{i=1}^{\lceil\sqrt{k}\ln{\frac{1}{1-\alpha}} \rceil}
 e(\lceil \beta_i k_i \rceil )^{\lceil \beta_i k_i \rceil +1}\Bigl(\frac{\sqrt{k}}{e}\Bigr)^{\lceil \beta_i k_i \rceil }e^{\frac{k_0+k_i}{\sqrt{k}}}\\
 ={}&\frac{(k\ln n)^{O(\sqrt{k})}}{k_1!}\cdot
 \Bigl(\prod_{i=1}^{\lceil\sqrt{k}\ln{\frac{1}{1-\alpha}}
 \rceil}(\lceil \beta_i k_i \rceil )^{\lceil \beta_i k_i \rceil
 +1}\Bigr)\cdot\\
 &\hskip 10em\cdot\Bigl(\frac{\sqrt{k}}{e}\Bigr)^{\sum_{i=1}^{\lceil\sqrt{k}\ln{\frac{1}{1-\alpha}} \rceil}\lceil \beta_i k_i \rceil }
 \cdot e^{k_0\ln\frac{1}{1-\alpha}+\frac{\sum_{i=1}^{\lceil\sqrt{k}\ln{\frac{1}{1-\alpha}} \rceil}k_i}{\sqrt{k}}}\\
 ={}&\frac{(k\ln n)^{O(\sqrt{k})}}{k_1!}\cdot \Bigl(\frac{1}{1-\alpha}\Bigr)^{k_0}\cdot
 \Bigl(\prod_{i=1}^{\lceil\sqrt{k}\ln{\frac{1}{1-\alpha}}
 \rceil}(\lceil \beta_i k_i \rceil )^{\lceil \beta_i k_i \rceil
 +1}\Bigr)\cdot\\
 &\hskip 10em\cdot\Bigl(\frac{\sqrt{k}}{e}\Bigr)^{\sum_{i=1}^{\lceil\sqrt{k}\ln{\frac{1}{1-\alpha}} \rceil}\lceil \beta_i k_i \rceil }
 \cdot e^{\frac{\sum_{i=1}^{\lceil\sqrt{k}\ln{\frac{1}{1-\alpha}} \rceil}k_i}{\sqrt{k}}}.
\end{align*}
And we know the sum of all $\lceil \beta_i k_i \rceil $ is $k_1$ by construction, the sum of all $k_i$ is upper bounded as follows
\begin{align*}
 \sum_{i=1}^{\lceil\sqrt{k}\ln{\frac{1}{1-\alpha}} \rceil}k_i&
 \le\sum_{i=1}^{\lceil\sqrt{k}\ln{\frac{1}{1-\alpha}} \rceil}k_1\Bigl(1-\frac{1-e^{\frac{-i}{\sqrt{k}}}}{\alpha}\Bigr)\\
 &=k_1\frac{\alpha-1}{\alpha}\lceil\sqrt{k}\ln{\frac{1}{1-\alpha}} \rceil 
 + \frac{k_1}{\alpha} \sum_{i=1}^{\lceil\sqrt{k}\ln{\frac{1}{1-\alpha}} \rceil}e^{\frac{-i}{\sqrt{k}}}\\
 & \le -k_1\frac{1-\alpha}{\alpha}\sqrt{k}\ln{\frac{1}{1-\alpha}} 
 + \frac{k_1}{\alpha}+k_1\sqrt{k}\;,
 \end{align*}
 where we used  \Cref{obs:useful}. The product of all $(\lceil \beta_i k_i \rceil )^{\lceil \beta_i k_i \rceil +1}$ can also be upper bounded using  \Cref{obs:useful} and the arithmetic sum formula, that is,
 \begin{align*}
  &\prod_{i=1}^{\lceil\sqrt{k}\ln{\frac{1}{1-\alpha}} \rceil}(\lceil \beta_i k_i \rceil )^{\lceil \beta_i k_i \rceil +1}
  =\prod_{i=1}^{\lceil\sqrt{k}\ln{\frac{1}{1-\alpha}} \rceil}
  \Bigl(\frac{\bigl(1+O \bigl(\frac{1}{\sqrt{k}}\bigr)\bigr)}{\alpha\sqrt{k}}e^{-\frac{i}{\sqrt{k}}}k_1\Bigr)
  ^{\lceil \beta_i k_i \rceil +1}\\
  &{}\le{}\Bigl(\frac{\bigl(1+O \bigl(\frac{1}{\sqrt{k}}\bigr)\bigr)}{\alpha\sqrt{k}}k_1\Bigr)
  ^{\sum_i \lceil \beta_i k_i \rceil  +1}
  \cdot
  e^{\sum_{i=1}^{\lceil\sqrt{k}\ln{\frac{1}{1-\alpha}} \rceil}\frac{-i}{\sqrt{k}}\cdot 
  \Bigl(\frac{k_1(1+O(1/\sqrt{k}))e^{\frac{-i}{\sqrt{k}}}}{\alpha \sqrt{k} }+1\Bigr)}\\
  &{}={}\Bigl(\frac{\bigl(1+O \bigl(\frac{1}{\sqrt{k}}\bigr)\bigr)}{\alpha\sqrt{k}}k_1\Bigr)
  ^{k_1+\lceil\sqrt{k}\ln{\frac{1}{1-\alpha}} \rceil-1}
  \cdot
  e^{-\frac{(1+O(1/\sqrt{k}))k_1}{\alpha k}\sum_i ie^{\frac{-i}{\sqrt{k}}}}
  \cdot e^{ \frac{-\sum_i i}{\sqrt{k}}}\\
  &{}\le{}\Bigl(\frac{\bigl(1+O \bigl(\frac{1}{\sqrt{k}}\bigr)\bigr)}{\alpha\sqrt{k}}k_1\Bigr)
  ^{k_1+\lceil\sqrt{k}\ln{\frac{1}{1-\alpha}} \rceil-1}
  \\[-5pt]
  &\hskip8em\cdot
  e^{-\frac{(1+O(1/\sqrt{k}))k_1}{\sqrt{k}\alpha}
  \left( -(1-\alpha)k\ln\frac{1}{1-\alpha}+\alpha (k+O(1))\right)}
  \cdot  
  e^{\frac{-(\lceil\sqrt{k}\ln{\frac{1}{1-\alpha}} \rceil+1)(\lceil\sqrt{k}\ln{\frac{1}{1-\alpha}} \rceil)}{2\sqrt{k}}}\\
  &{}\le{}\Bigl(\frac{\bigl(1+O \bigl(\frac{1}{\sqrt{k}}\bigr)\bigr)}{\alpha\sqrt{k}}k_1\Bigr)
  ^{k_1+\lceil\sqrt{k}\ln{\frac{1}{1-\alpha}} \rceil-1}
  \cdot
  \Bigl(\frac{1}{1-\alpha}\Bigr)^{\frac{(1-\alpha)k_1}{\alpha}(1+O(1/\sqrt{k}))}
  \\
  &\hskip8em\cdot e^{-k_1\sqrt{k}(1+O(\frac{1}{\sqrt{k}}))}
  \cdot \Bigl(\frac{1}{1-\alpha}\Bigr)^{-\sqrt{k}-1/2}\;.
 \end{align*}
 If we substitute, and using again the Stirling inequalities, 
 the total number of functions in the mapping family is at most
   \begin{align*} 
 &\frac{(k \ln n)^{O(\sqrt{k})}}{k_1!}\cdot \Bigl(\frac{1}{1-\alpha}\Bigr)^{k_0}\cdot
 \Bigl(\frac{\sqrt{k}}{e}\Bigr)^{k_1}
 \cdot e^{-k_1\frac{1-\alpha}{\alpha}\ln{\frac{1}{1-\alpha}}
 + \frac{k_1}{\alpha\sqrt{k}}+k_1}\\
 &\cdot \Bigl(\frac{(1+O \bigl(\frac{1}{\sqrt{k}}\bigr)}{\alpha\sqrt{k}}k_1\Bigr)
  ^{k_1+\lceil\sqrt{k}\ln{\frac{1}{1-\alpha}} \rceil-1}
 \cdot \Bigl(\frac{1}{1-\alpha}\Bigr)^{-\sqrt{k}-1/2+\frac{1-\alpha}{\alpha\sqrt{k}}k_1}
 \cdot e^{-k_1\sqrt{k}(1+O(\frac{1}{\sqrt{k}}))}\\
 = &
 (k \ln n)^{O(\sqrt{k})}\cdot \Bigl(\frac{1}{1-\alpha}\Bigr)^{k_0}\cdot k_1^{-k_1+1}
 \cdot \Bigl(\frac{\sqrt{k}}{e}\Bigr)^{k_1}
 \cdot \Bigl(\frac{1+O \bigl(\frac{1}{\sqrt{k}}\bigr)}{\alpha\sqrt{k}}k_1\Bigr)^{k_1-1}\\
  = &
 (k \ln n)^{O(\sqrt{k})}\cdot \Bigl(\frac{1}{1-\alpha}\Bigr)^{k_0}
 \cdot \sqrt{k}^{k_1}\cdot e^{-k_1}
 \cdot \Bigl(\frac{1+O \bigl(\frac{1}{\sqrt{k}}\bigr)}{\alpha\sqrt{k}}\Bigr)^{k_1-1}\\
 = &
 (k \ln n)^{O(\sqrt{k})}\cdot \Bigl(\frac{1}{1-\alpha}\Bigr)^{k_0}
 \cdot e^{-k_1}
 \cdot \Bigl(\frac{1+O(\frac{1}{\sqrt{k}})}{\alpha}\Bigr)^{k_1-1}\\
 =& (k\ln n)^{O(\sqrt{k})}\cdot \Bigl( \frac{1}{1-\alpha}\Bigr)^{k_0}\cdot \Bigl(\frac{1}{\alpha}\Bigr)^{k_1}\;,
 \end{align*}
 as we wanted.
 
 In order to construct the mapping family required in each step, we need $nk^{O(k^{5/2})}$ time, there are
 $O(\sqrt{k})$ mapping families, and we can combine them in linear time, achieving the required time bound.
\end{proof}

For the proof of \Cref{lem:bisn2} we need the following results:
\begin{obs}\label{obs:useful}
 Some useful calculations for  \Cref{lem:bisn2}:
 \begin{align*}&\sum_{i=1}^{\lceil\sqrt{k}\ln\frac{1}{1-\alpha}\rceil} e^{\frac{-i}{\sqrt{k}}}\le
 1+\alpha\sqrt{k}\;,\\
 &\sum_{i=1}^{\lceil\sqrt{k}\ln\frac{1}{1-\alpha}\rceil} i e^{\frac{-i}{\sqrt{k}}}\ge
 -(1-\alpha)k\ln\frac{1}{1-\alpha}+\alpha (k+O(1))\;.\end{align*}
\end{obs}
\begin{proof}
 For the first inequality we use the geometric sum and $e^{-\frac{1}{\sqrt{k}}}\le 1$:
 \begin{align*}
 \sum_{i=1}^{\lceil\sqrt{k}\ln\frac{1}{1-\alpha}\rceil} e^{\frac{-i}{\sqrt{k}}}
 &=e^{-\frac{1}{\sqrt{k}}}\frac{1-e^{-\frac{\lceil\sqrt{k}\ln\frac{1}{1-\alpha}\rceil}{\sqrt{k}}}}{1-e^{-\frac{1}{\sqrt{k}}}}\\
 &\le e^{-\frac{1}{\sqrt{k}}}\frac{1-e^{-\frac{\sqrt{k}\ln\frac{1}{1-\alpha}+1}{\sqrt{k}}}}{1-e^{-\frac{1}{\sqrt{k}}}}\\
 &\le\frac{1-\left(\frac{1}{1-\alpha}\right)^{-1}e^{-\frac{1}{\sqrt{k}}}}{1-e^{-\frac{1}{\sqrt{k}}}}\\
 &=\frac{1-(1-\alpha)e^{-\frac{1}{\sqrt{k}}}}{1-e^{-\frac{1}{\sqrt{k}}}}\\
 &=1+\frac{\alpha\sqrt{k}}{(1+O(1/\sqrt{k}))}\\
 &\le 1+\alpha\sqrt{k}\;.
 \end{align*}
 For the second inequality we also have a closed formula and $e^{-\frac{1}{\sqrt{k}}}/(1-e^{-\frac{1}{\sqrt{k}}})^2=k+O(1)$, which we can simplify:
 \begin{align*}
 \sum_{i=1}^{\lceil\sqrt{k}\ln\frac{1}{1-\alpha}\rceil} ie^{\frac{-i}{\sqrt{k}}}
 &=\frac{e^{-\frac{1}{\sqrt{k}}}\left(\lceil\sqrt{k}\ln\frac{1}{1-\alpha}\rceil 
 e^{-\frac{\lceil\sqrt{k}\ln\frac{1}{1-\alpha}\rceil+1}{\sqrt{k}}}
 -\left(\lceil\sqrt{k}\ln\frac{1}{1-\alpha}\rceil+1\right)e^{-\frac{\lceil\sqrt{k}\ln\frac{1}{1-\alpha}\rceil}{\sqrt{k}}}+1\right)}{\left(1-e^{-\frac{1}{\sqrt{k}}}\right)^2}\\
 &=\frac{e^{-\frac{1}{\sqrt{k}}}\left(
 \lceil\sqrt{k}\ln\frac{1}{1-\alpha}\rceil e^{-\frac{\lceil\sqrt{k} \ln\frac{1}{1-\alpha}\rceil}{\sqrt{k}}}
 (e^{-\frac{1}{\sqrt{k}}}-1)-e^{-\frac{\lceil\sqrt{k}\ln\frac{1}{1-\alpha}\rceil}{\sqrt{k}}}+1\right)}{\left(1-e^{-\frac{1}{\sqrt{k}}}\right)^2}\\
 &=-\frac{e^{-\frac{1}{\sqrt{k}}}\left(
 \lceil\sqrt{k}\ln\frac{1}{1-\alpha}\rceil e^{-\frac{\lceil\sqrt{k} \ln\frac{1}{1-\alpha}\rceil}{\sqrt{k}}}\right)}{1-e^{-\frac{1}{\sqrt{k}}}}+e^{-\frac{1}{\sqrt{k}}}
 \frac{1-e^{-\frac{\lceil\sqrt{k}\ln\frac{1}{1-\alpha}\rceil}{\sqrt{k}}}}{\left(1-e^{-\frac{1}{\sqrt{k}}}\right)^2}\;.\end{align*}
 Which means that 
 \begin{align*}
 \sum_{i=1}^{\lceil\sqrt{k}\ln\frac{1}{1-\alpha}\rceil} ie^{\frac{-i}{\sqrt{k}}}
 &\ge-\frac{e^{-\frac{1}{\sqrt{k}}}\left(
 \sqrt{k}\ln\frac{1}{1-\alpha} e^{-\frac{\sqrt{k} \ln\frac{1}{1-\alpha}}{\sqrt{k}}}\right)}{1-e^{-\frac{1}{\sqrt{k}}}}+e^{-\frac{1}{\sqrt{k}}}
 \frac{1-e^{-\frac{\sqrt{k}\ln\frac{1}{1-\alpha}}{\sqrt{k}}}}{\left(1-e^{-\frac{1}{\sqrt{k}}}\right)^2}\\ 
 &\ge-\frac{\left(\sqrt{k}\ln\frac{1}{1-\alpha}\right)(1-\alpha)}{\frac{1}{\sqrt{k}}\bigl(1+O(1/\sqrt{k})\bigr)}
 +\frac{ e^{-\frac{1}{\sqrt{k}}}(1-(1-\alpha))}{\left(1-e^{-\frac{1}{\sqrt{k}}}\right)^2}\\
 &=-\frac{k\ln\left(\frac{1}{1-\alpha}\right)(1-\alpha)}{\bigl(1+O(1/\sqrt{k})\bigr)}
 +\alpha (k+O(1))  \\
 &\ge -(1-\alpha)k\ln\frac{1}{1-\alpha}+\alpha (k+O(1))\;.
 \end{align*}
 Thus, proving the desired results.
\end{proof}

\begin{lemma}\label{lem:bisfastsmalln}
 Let $0< \alpha\le 1/2$, and let $n$ and $k_0,k_1$ be integers with $n\ge 4k^6$,
 where $k=k_0+k_1$ and $n=k^{O(1)}$.
 We can construct an $(n,k_0,k_1,\alpha,1)$\mapfam of size 
 $\bigl(\frac{1}{\alpha}\bigr)^{k}k^{O(k^{5/6})}$ in linear time.
\end{lemma}

\begin{proof}
 Given a set $S_0\subseteq[n]$ with $|S_0|=k_0$ we can guess $\ell_0=\lceil k_0^{2/3}\rceil$ 
 disjoint intervals $I^0_1,\hdots, I^0_{\ell_0}$ such that each interval 
 contains at most 
 $\lceil k_0^{1/3}\rceil $ elements from $S_0$. 
 Given a set $S_1\subseteq[n]$ one can 
 also guess \smash{$\ell_1=\lceil k_1^{2/3}\rceil$} disjoint intervals $I^1_1,\hdots, I^1_{\ell_1}$
 each containing at most $\lceil k_1^{1/3}\rceil$ elements from $S_1$. Observe, that the 
 intervals $I^0_i$ might not be disjoint with the intervals
 \smash{$I^1_j$}, but we can take the 
 intervals $I_1,\hdots, I_{t}$ with $t\le \ell_0+\ell_1$ given by taking
 all of the interval delimiters and making a separate interval every
 time we hit a new delimiter.
 
 We can also guess now an interval of size $2k^5$ that does not contain any element from $S_0$
 or $S_1$. We build now $t$ extra sets $I'_i$ of size $n'_i\ge k^4$
 by combining $I_i$ with $k^4$ elements from the
 large interval. In that way we obtain sets $I'_i$ containing $k^i_0\le\lceil k_0^{1/3}\rceil$
 elements from $S_0$ and $k^i_1\le\lceil k_1^{1/3}\rceil$ elements from $S_1$. 
 We can construct,
 using  \Cref{lem:bisn2}, for each $I'_i$ an $(n'_i,k^i_0,k_1^i,\alpha,1)$\mapfam of size
 $\bigl(\frac{1}{1-\alpha}\bigr)^{k^i_0}\bigl(\frac{1}{\alpha}\bigr)^{k^i_1}k^{O(k^{1/6})}$
 in time $n_i'k^{O(k^{5/6})}=k^{O(k^{5/6})}$. Combining all of the families into one yields
 a family of size $\bigl(\frac{1}{1-\alpha}\bigr)^{k_0}\bigl(\frac{1}{\alpha}\bigr)^{k_1}k^{O(k^{5/6})}$, which is in the worst case
 $\bigl(\frac{1}{\alpha}\bigr)^{k}k^{O(k^{5/6})}$,
 in linear time.
\end{proof}
We can use the same  \Cref{lem:extendn} but for a general $\alpha$ to
extend  \Cref{lem:bisfastsmalln} and get a mapping family
with general $n$ also in linear time.

\begin{lemma}\label{thm:bisfastgen}
 Let $0< \alpha\le 1/2$ be a constant, 
 and let $n$ and $k_0,k_1$ be integers with $n\ge 4k^6$,
 where $k=k_0+k_1$.
 We can construct an $(n,k_0,k_1,\alpha,1)$\mapfam of size 
 $\bigl(\frac{1}{\alpha}\bigr)^{k}k^{O(k^{5/6})}\log n$ in linear time.
\end{lemma}

\begin{proof}
 We have to do the equivalent procedure that we did in \Cref{lem:extendn} to a  
 mapping family constructed using  \Cref{lem:bisfastsmalln}.
 We take a uniform $(n,k,5k^6)$-splitter of size $O(k^6\log n)$
 constructed using \Cref{thm:k3splitter} 
 and combine a function $f\colon[n]\to[m]$ where $5k^6\ge m\ge 4k^6$ of this splitter with an $(m,k_0,k_2,\alpha,1)$\mapfam constructed using \Cref{lem:bisfastsmalln}. 
 Let $g$ be one of the combined functions, this function 
 might not be perfectly balanced. Similarly as in the proof of
 \Cref{lem:extendn}, we can see that if $g$ is balanced enough, we can find $k+1$
 different functions $g'$ such that each possible $k$-subset is respected by one of them and $|g'^{-1}(1)|=\lceil\alpha n \rceil$ as we want. We can do this because
 $\alpha\ge 1/\sqrt{k}$, so the deviation will be even smaller than it was in \Cref{lem:extendn}.
 The size is the product of the splitter and mapping family sizes.
\end{proof}

Observe that in particular, if we apply \Cref{thm:bisfastgen} to obtain an $(n,k,0,\alpha,1)$\mapfam this is the same as an $(n,k,\alpha)$-bisector. We could ask ourselves then, why would we need all of the results of \Cref{sec:bisect}, but
with a closer look we see that aside from being able to construct bisectors for 
smaller values of $n$, the size of the bisectors built in \Cref{sec:bisect} is
completely independent of $n$. This is not the case in \Cref{sec:usets}.
We can finally extend the result of \Cref{thm:bisfastgen} to uniform universal sets.

\begin{theorem}\label{lem:busslow} 
Let $0\le \alpha\le 1/2$, and let $n$ and $k$ be integers with $n\ge 4k^6$.
We can construct a uniform $(n,k,\alpha)$-universal set of size $\bigl(\frac{1}{\alpha}\bigr)^{k}k^{O(k^{5/6})}\log n$ in linear time .
\end{theorem}

\begin{proof}
 One can take for every value of $k_0$ and $k_1$ the $(n,k_0,k_1,\alpha,1)$\mapfam given by
 \Cref{thm:bisfastgen}. 
 Their union is a uniform universal set because for every pair of sets $S_0$ and $S_1$, 
 we can take a function from the mapping family with the appropriate values of $k_0$ and $k_1$,
 which will assign all the elements of $S_0$ to 0 and $S_1$ to 1. 
 
 There are $k+1$ choices for the values of $k_0$ and $k_1$, which means that the size 
 of all of these mapping families is only a factor of $k+1$ bigger than the size of one of them.
\end{proof} 

\section{Conclusion}

While our uniform splitters are built in a very simple way, the sizes
of $(n,k,k^3)$-splitters by Naor et al. are smaller than ours and they
obtain also small $(n,k,k^2)$-splitters, which we do not have.
Closing this gap remains an open question.

\bibliographystyle{plainurl}
\bibliography{bisect}
%

\end{document}